\documentclass[titlepage]{amsart}
\usepackage{amsaddr}
\usepackage{amsthm,amsmath,amssymb,eucal}
\usepackage{geometry}
\usepackage{resizegather}
\usepackage{graphicx}
\usepackage{xcolor}
\usepackage[normalem]{ulem}
\usepackage[latin1]{inputenc}
\usepackage{enumerate}
\usepackage{hyperref}

\usepackage{lineno,hyperref}
\usepackage{amsthm,amsmath,amssymb,eucal}
\usepackage{caption}
\usepackage{subcaption}
\usepackage{color,colordvi}

\modulolinenumbers[5]

\newcommand{\ee}{\mathrm{e}}

\newcommand{\R}{\mathbb{R}}
\newcommand{\OO}{\mathcal{O}}

\newtheorem{claim}{Claim}[section]

\newtheorem{proposition}[claim]{Proposition}

\newtheorem{remarks}[claim]{Remarks}

\begin{document}

\title{Quantum graph resonances by cut-off technique}

\author{Pavel Exner}
\address{{Doppler Institute for Mathematical Physics and Applied Mathematics, Czech Technical University,
B\v rehov{\'a} 7, 11519 Prague, Czechia} {\rm and} {Department of Theoretical Physics, Nuclear Physics Institute, Czech Academy of Sciences, 25068 \v{R}e\v{z} near Prague, Czechia}
}
\email{exner@ujf.cas.cz}

\author{Ji\v{r}\'{\i} Lipovsk\'{y}}
\address{{Department of Physics, Faculty of Science, University of Hradec Kr\'alov\'e, Rokitansk\'eho 62, 50003 Hradec Kr\'alov\'e, Czechia}
}
\email{jiri.lipovsky@uhk.cz}

\author{Jan Peka\v{r}}
\address{{Institute of Theoretical Physics, Faculty of Mathematics and Physics, Charles University, V Hole\v{s}ovi\v{c}k\'ach 2, 18040 Prague, Czechia}
}
\email{honzapekar28@gmail.com}

\begin{abstract}
We demonstrate how resonances in a quantum graph consisting of a compact core and semi-infinite leads can be identified from the eigenvalue behavior of the cut-off system.
\end{abstract}

\maketitle

%%%%%%%%%%%%%%%%%%%%%%
\section{Introduction}
%%%%%%%%%%%%%%%%%%%%%%

Quantum graphs have a long history; the concept was proposed by Linus Pauling \cite{Pa36} in the 1930s but then happily forgotten for about half a century. It was rediscovered in connection with progress of solid-state fabrication techniques and it flourishes since then being useful not only as a model of numerous nanostructures, but what is probably more important, also as a versatile tool to address various theoretical questions in quantum mechanics; we refer to the monographs \cite{PPP05, BK13, KN22, Ku24} for introduction to the field and an extensive bibliography.

The second notion in the title, the resonances, concerns a wide family of effects the importance of which certainly need not be advocated. There are various methods to identify resonances in quantum systems. Two most frequently appearing ones are finding the analytical continuation of the Hamiltonian resolvent explicitly, as is the case in the classical Friedrichs model \cite{Fr48} and its numerous generalizations, see \cite{Ex13} for a review, and the complex scaling method \cite{AC71, BC71} with application particularly in molecular physics \cite{Si79}.

One of the main mechanisms of resonance production, as known already from the Friedrichs model, are perturbation of eigenvalues embedded in the continuous spectrum. Such resonances are ubiquitous in quantum graphs where the unique continuation property does not hold \cite[Sec.~3.4]{BK13}, and as a result, one can have compactly supported eigenfunctions which 'do not talk' to the rest of the graph, producing then resonances by geometric perturbations \cite{EL10}. In this connection it is useful to recall that there are different resonance notions, in particular, the resolvent and scattering ones. The two are not the same in general, the former being the property of the Hamiltonian alone, while the latter involves comparison of the full Hamiltonian to the free one. However, in many situations the two notions coincide -- naturally up to the set of embedded eigenvalues not influenced by the scattering -- which is also the case for quantum graphs \cite{EL07, Lip16}.

To derive the resonance condition in this case may require, especially for more complicated graphs, a not so simple algebra, hence alternative methods could be useful. In this connection an approach popular in physics, usually dubbed `resonance in a box' \cite{MCD80, ZMZ09} or `stabilization' \cite{HT70, DS21}, comes to mind, in which the search for complex poles is replaced by spectral analysis of the system confined to a bounded region and an inspection of the dependence of the eigenvalue behavior on the size of this `box'. The intuition behind this method is clear, however, in most cases a mathematically convincing justification is missing; the only case where the connection is rigorously established is the one-dimensional potential scattering \cite{HM00}. Our aim here is to demonstrate how quantum graph resonances can be identified using such a technique.

%%%%%%%%%%%%%%%%%%%%%%%%%%%%%%%%%%%%%%%
\section{Preliminaries}\label{s:prelim}
%%%%%%%%%%%%%%%%%%%%%%%%%%%%%%%%%%%%%%%

Let us first recall the problem setting. Quantum graph is a common shorthand for the Laplacian, or more generally a Schr\"odinger operator, on a metric graph, which we denote as $\Gamma$. The latter consists of a set of vertices, $\mathcal{V} = \{v_j\}$, a set $\mathcal{E} = \{e_i:\, i= 1, \dots, N\}$ of finite edges of lengths $l_i>0$ connecting the elements of $\mathcal{V}$, together with a set of semi-infinite edges, $\mathcal{E}_{\infty} = \{e_{\infty, j}:\, j = 1, \dots, M\}$, often called \emph{leads}, attached to some or all of the vertices. The edges are identified with line segments which allows us to consider differential operators on $\Gamma$. For the sake of simplicity we suppose here that external fields are absent so that the dynamics on $\Gamma$ is governed by the Hamiltonian which acts as $-\frac{d^2}{dx^2}$ on each edge; we neglect physical constants the values of which play no role in our considerations. To make such an operator self-adjoint, one has to impose appropriate condition matching the wave function at the vertices. It is well known that the most general one are
 % ------------- %
%\begin{subequations} \label{matching}
\begin{equation}\label{Vert_con_ind}
		(U_j-I){\psi}_j(v)+i(U_j+I){\psi^{'}_j}(v)=0,
\end{equation}
 % ------------- %
at the vertex $v_j$, where $\psi_j$ and $\psi'_j$ are vectors of the one-sided limits of wave function values and their derivatives, respectively, at the vertex $v_j$, and $U$ is a $\mathrm{deg}\,v_j \times \mathrm{deg}\,v_j$ unitary matrix. These conditions are usually attributed to Kostrykin and Schrader \cite{KS99} but in fact they were known already to Rofe-Beketov \cite{RB69}. The operator domain then consists of elements of the Sobolev space $W^{2,2}(\mathcal{E} \bigoplus \mathcal{E}_\infty)$ satisfying condition \eqref{Vert_con_ind}.

It is important for our discussion that we do not need to specify the topology of $\Gamma$; we may assume that all the vertices are merged to produce a `flower' graph \cite{EL10} with $N$ `petals' \cite[Sec.~1.4.6]{BK13} and $M$ semi-infinite `stems', all meeting at a single point, where the condition
 % ------------- %
\begin{equation}\label{Gen_vert_con_flower}
		(U-I){\Psi}(v)+i(U+I){\Psi^{'}}(v)=0
\end{equation}
%\end{subequations}
 %----------------%
with a $(2N+M)\times(2N+M)$ unitary matrix $U$ is imposed; the local character of the wave function matching is restored by choosing $U$ of a block-diagonal form in which each block corresponds to a vertex of the original graph $\Gamma$.

One way to find the resonances is to use an exterior complex scaling which acts on the leads leaving the compact core of $\Gamma$ intact. We have to choose appropriately the wave function Ansatz. Since we assume that the dynamics is governed by the Laplacian, we use $f_i(x) = a_i\sin{kx} + b_i\cos{kx},\, i=1,\dots,N$, while on the external ones we have plane waves which behave as $g_{j\theta}(x) = \ee^{\theta/2}g_j(xe^{\theta}),\, j=1,\dots,M$ under the scaling \cite{EL07}, so that the corresponding boundary values are $g_j(0) = \ee^{-\theta/2}g_{j\theta}$ and $g^{'}_j(0) = ik\ee^{-\theta/2}g_{j\theta}$, respectively. Substituting into \eqref{Gen_vert_con_flower}, we get a system of $2N+M$ equations for the coefficients $a_i$, $b_i$ and $\ee^{-\theta/2}g_{j\theta}$ the solvability requirement of which is by \cite{EL07} equivalent to the condition
 %----------------%
\begin{equation}\label{Resonance_det}
    \det[(U-I)C_1(k)+ik(U+I)C_2(k)] = 0
\end{equation}
 %----------------%
where $U$ is the vertex condition matrix of \eqref{Gen_vert_con_flower} with the entries appropriately rearranged and the matrices $C_1$ and $C_2$ are given by the transfer matrices on the finite edges,
 %----------------%
\begin{itemize}
    \item $C_1 = \begin{pmatrix}
                J_{2N\times2N}(k) & 0\\
			0 & I_{M\times M} \\
		\end{pmatrix}$, where $I_{M\times M}$ is the unit matrix and the $J(k)$ part is block-diagonal consisting of $2\times 2$ matrices $J_i(k)= \begin{pmatrix}
                0 & 1\\
			\sin{kl_i} & \cos{kl_i}\\
            \end{pmatrix}$,
     \item  $C_2 = \begin{pmatrix}
                K_{2N\times2N}(k) & 0\\
			0 & iI_{M\times M} \\
		\end{pmatrix}$, where $K(k)$ has a similar structure with the $2\times 2$ blocks of the form $K_i(k) =  \begin{pmatrix}
                1 & 0\\
			-\cos{kl_i} & \sin{kl_i}\\
            \end{pmatrix}$.
\end{itemize}
 %----------------%
The problem can be further simplified to a system of $2N\times 2N$ equations by replacing $U$ with an effective energy dependent matrix $\tilde{U}(k)$ referring the compact core of $\Gamma$ alone \cite{EL10} but we will not do that because the main idea here is to replace the leads by finite edges changing naturally the `exterior' solutions $g_j$. With that in mind we note that, due to the choice of our Ansatz, the imaginary unit in the lower right block of $C_2$ is nothing but $\frac{1}{k}\frac{g^{'}_j(0)}{g_j(0)}$.

As a warm-up, let us now look at three examples of graph resonances two which were investigated in detail in \cite{EL10}. In the first one $\Gamma$ is a loop with two leads, sketched in Figure~\ref{Loop_graph}; we have two internal edges of lengths $l_1 = l(1-\lambda)$, $l_2 = l(1+\lambda)$ with the parameter $\lambda \in [0,1]$, coupled to two half-lines, each endpoint being connected to a lead. For the vertex conditions we choose a four-parameter family \cite{ES89} with the wave function continuous on the loop and tunable coupling to the leads,
 %----------------%
\begin{equation*}
    \begin{aligned}
    f_1(0) &= f_2(0),\; f_1(l_1) = f_2(l_2),\\
    f_1(0) &= \alpha^{-1}_1(f^{'}_1(0)+f^{'}_2(0)) + \gamma_1g^{'}_1(0),\\
    f_1(l_1) &= -\alpha^{-1}_2(f^{'}_1(l_1)+f^{'}_2(l_2)) + \gamma_2g^{'}_2(0),\\
    g_1(0) &= \overline{\gamma}_1(f^{'}_1(0)+f^{'}_2(0)) + \tilde{\alpha}^{-1}_1g^{'}_1(0),\\
    g_2(0) &= -\overline{\gamma}_2(f^{'}_1(l_1)+f^{'}_2(l_2)) + \tilde{\alpha}^{-1}_2g^{'}_2(0).
    \end{aligned}
\end{equation*}
 %----------------%
	\begin{figure}[b]\centering
		\includegraphics[width=.7\textwidth]{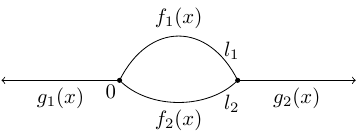}
		\caption{The loop graph with two external leads.}
		\label{Loop_graph}
	\end{figure}
 %----------------%
The spectrum of the corresponding Hamiltonian is continuous covering $[0,\infty)$; depending on the coupling coefficients, it may have up to eight negative eigenvalues having a common symmetric restriction of deficiency indices $(8,8)$ with the graph in which the leads are decoupled from the loop \cite[Thm.~8.19]{We80}. Resonances come from embedded eigenvalues which occur when $l_1$ and $l_2$ are commensurate and the Laplacian on the loop has an eigenfunction with zeros at the vertices. The corresponding resonance condition \cite[eq.~(15)]{EL10} is
 %----------------%
\begin{equation} \label{loop-res}
    \sin{kl(1-\lambda)}\sin{kl(1+\lambda)}-4k^2\beta^{-1}_1(k)\beta^{-1}_2(k)\sin^2{kl} + k[\beta^{-1}_1(k) + \beta^{-1}_2(k)]\sin{2kl} = 0,
\end{equation}
 %----------------%
where $\beta^{-1}_j(k) := \alpha^{-1}_j + \frac{ik|\gamma_j|^2}{1-ik\tilde{\alpha}^{-1}_j}$. Next we replace $\Gamma$ with a finite graph which instead of the leads has edges of length $L$ assuming that at the `loose' endpoints Dirichlet condition is imposed. The spectrum of the Laplacian changes completely; it is now purely discrete. It is not difficult to find the eigenvalues using condition \eqref{Gen_vert_con_flower} with the Ansatz on the cut leads replaced by $g_j(x) = c_j\, \sin k(L-x)$; now, of course, no complex scaling is needed. The resulting spectral condition has the form \eqref{loop-res} again, however, with modified coefficients $\beta^{-1}_j(k)$, now equal to $\alpha^{-1}_j + \frac{-k|\gamma_j|^2\cot{kL}} {1+k\tilde{\alpha}^{-1}_j\cot{kL}}$, in other words, the imaginary unit in the previous expression of $\beta^{-1}_j(k)$ is replaced by $-\cot{kL}$.

The second example is a cross-shaped graph $\Gamma$ with two internal edges of lengths $l_1 = l(1-\lambda)$ and $l_2 = l(1+\lambda)$, the parameter $\lambda$ running again through $[0,1]$, and Dirichlet conditions imposed at their endpoints, and two semi-infinite leads, meeting in a single vertex, see Figure~\ref{Cross_graph}. The central vertex condition is now simpler; we choose it as the standard $\delta$ coupling with a strength parameter $\alpha\in\R$.
 %----------------%
	\begin{figure}[b]\centering
		\includegraphics[width=.7\textwidth]{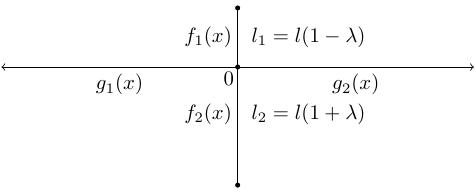}
		\caption{The cross-shaped graph.}
		\label{Cross_graph}
	\end{figure}
 %----------------%
The resonance condition of this system \cite[eq.~(19)]{EL10} is
 %----------------%
 \begin{equation}\label{Cross_cond_res}
     2k\sin{2kl} + (\alpha -2ik)(\cos{2kl\lambda}-\cos{2kl}) = 0.
 \end{equation}
 %----------------%
On the other hand, the spectral condition of a cut-off cross-shaped graph with two `horizontal' edges of length $L$ and Dirichlet endpoints is easily seen to be
 %----------------%
  \begin{equation}\label{Cross_cond_spec}
     2k\sin{2kl} + (\alpha +2k\cot{kL})(\cos{2kl\lambda}-\cos{2kl}) = 0,
 \end{equation}
 %----------------%
that is, the imaginary unit of \eqref{Cross_cond_res} has been again changed into $-\cot{kL}$. Again, for large enough negative $\alpha$ the graph can also have a negative eigenvalue solving condition \eqref{Cross_cond_res} with the trigonometric functions replaced by hyperbolic ones.

The third example is a simple modification of the previous one in which $\Gamma$ is a T-shaped tree graph with three edges, the first two of lengths $l_1$ and $l_2$ with Dirichlet condition at the `loose' endpoints, the third one being a semi-infinite lead; at the vertex we again assume a $\delta$ coupling of strength $\alpha$. The resonance condition of this system is
 %----------------%
 \begin{equation*}
     \big(\textstyle{\frac{\alpha}{k}}-i\big)\sin{kl_1}\sin{kl_2} + \sin{kl_1}\cos{kl_2} + \sin{kl_2}\cos{kl_1} = 0,
 \end{equation*}
 %----------------%
while the spectral condition for the graph with the Dirichlet cut-off lead is
 \begin{equation*}
     \big(\textstyle{\frac{\alpha}{k}}+\cot{kL}\big)\sin{kl_1}\sin{kl_2} + \sin{kl_1}\cos{kl_2} + \sin{kl_2}\cos{kl_1} = 0.
 \end{equation*}
 %----------------%

%%%%%%%%%%%%%%%%%%%%%%%%%%%%%%%%%%%%%%%%%%%%%%%%%%%%%%%%%%%%%%%%%%%%%%%%%%
\section{Relating the cut-off graph spectra to resonances}\label{s:result}
%%%%%%%%%%%%%%%%%%%%%%%%%%%%%%%%%%%%%%%%%%%%%%%%%%%%%%%%%%%%%%%%%%%%%%%%%%

Let us first observe that the correspondence we have seen in the examples is valid generally:
 %----------------%
\begin{proposition}
Let $\Gamma$ consist of a compact part with $N$ edges and $M$ semi-infinite leads, and let $\Gamma_L$ be obtained from $\Gamma$ by cutting the leads to length $L$. The condition determining positive eigenvalues of the Laplacian on $\Gamma_L$ with coupling at $\mathcal{V}$ given by a matrix $U$ and Dirichlet conditions at the endpoints of the cut-off leads is of the form \eqref{Resonance_det} with the imaginary unit in $C_2$ replaced by $-\cot{kL}$.
\end{proposition}
 %----------------%
\begin{proof}
As in the examples, we use the original Ansatz on the compact core of $\Gamma$ while on the cut-off edges we replace the exponential functions $g_j$ with  $g_j(x) = c_j\, \sin k(L-x),\, j=1,\dots,M$ in which case the Dirichlet-to-Neumann factor $\frac{g^{'}_j(0)}{g_j(0)}=-k\cot kL$ enters the equation in place of $ik$.
\end{proof}

 %----------------%
\begin{remarks}
{\rm (a) The claim modifies easily if one uses a different boundary condition at the cut-off endpoints. Choosing, for instance, Neumann, one changes the exterior Ansatz to $g_j(x) = c_j\, \cos k(L-x),\, j=1,\dots,M$, and the imaginary unit in \eqref{Resonance_det} has to be replaced by $\tan{kL}$.} \\[.2em]
{\rm (b) We are interested in the positive eigenvalues of the cut-off graph Laplacian which we relate to the resonances of the original problem. As we have mentioned in the examples, however, for some vertex coupling matrices $U$ the graph with the leads can also have negative eigenvalues, $-\kappa^2$, which can be found using the Ansatz $g_j(x)=c_j\ee^{-\kappa x}$ on the leads. In a similar way, one can choose $g_j(x)=c_j\sinh\kappa(L-x)$ to determine eigenvalues of the cut-off graph. It is no surprise that the latter converge to those of the `full' graph as $L\to\infty$, the role of the imaginary unit in \eqref{Resonance_det} being now played by $-1$ and $-\cot kL$ replaced by $-\coth\kappa L$. By Dirichlet bracketing \cite[Sec.~XIII.15]{RS78}, the negative cut-off graph eigenvalues are decreasing with respect to $L$; note that they may not exist if $L$ is too small.}
\end{remarks}
 %----------------%

Having reached this conclusion, we can turn to the question about relations between the resonances of the spectrum of the cut-off graphs. Let us denote the left-hand side of the condition \eqref{Resonance_det} as $F(k)$. The examples above are simple in the sense that the imaginary unit in the resonance condition was present only in the zeroth and first power and as such the transition to the spectral condition is straightforward. But this does not need to be the case - consider again the example of cross-shaped resonator, but now with the central vertex coupling \eqref{Vert_con_ind} referring to
 %----------------%
\begin{equation*}
 U = \begin{pmatrix}
                \frac{3-i}{4}&\frac{-1-i}{4}&\frac{-1-i}{4}&\frac{-1-i}{4}\\[.3em]
			 \frac{-1-i}{4}&\frac{3-i}{4}&\frac{-1-i}{4}&\frac{-1-i}{4}\\[.3em]
              \frac{-1-i}{4}&\frac{-1-i}{4}&\frac{3-i}{4}&\frac{-1-i}{4}\\[.3em]
               \frac{-1-i}{4}&\frac{-1-i}{4}&\frac{-1-i}{4}&\frac{3-i}{4}
		\end{pmatrix}.
\end{equation*}
 %----------------%
Its resonance condition given by \eqref{Resonance_det} reads
 %----------------%
\begin{equation*}
    \sin k(l_1+l_2) + 4k \cos{kl_1} \cos{kl_2}+2i\cos{kl_1} \cos{kl_2} =0,
\end{equation*}
 %----------------%
but in distinction to our previous observations the spectral condition of the cut-off graph is
 %----------------%
\begin{equation*}
   \cot^2{kL}\,(\sin k(l_1+l_2) + 4k \cos{kl_1} \cos{kl_2})+2\cot{kL}\cos{kl_1} \cos{kl_2} =0.
\end{equation*}
 %----------------%
One can construct other examples, even having the $i$ present in more than two distinct powers, but it is not difficult to see that the correspondence based on the exchange of the imaginary unit for $-\cot{kL}$ is one-to-one regardless of the specific structure of the resonance condition, and \emph{a priori} we do not have a way to determine the power of $i$ in the individual terms of the resonance condition except for the fact that in the real part $\mathrm{Re}\,F(k)$ the powers are even, while in the imaginary part $\mathrm{Im}\,F(k)$ they are odd. Still, we are able to discern a pattern. Let us rewrite the general spectral condition of the cut-off graph $\tilde{F}(k)$ as
 %----------------%
\begin{equation*}
    \tilde{F}(k) =\sum_{j =0}^n(-\cot{kL})^{j}c_j(k) =0
\end{equation*}
 %----------------%
with the appropriate $n$, where the coefficients $c_j(k)$ consist of all the terms in the spectral condition associated with the power of $(-\cot{kL})^{j}$. In a similar fashion, we can rewrite the resonance condition of the original graph as
 %----------------%
\begin{equation}\label{Resonance_cond_expanded}
    F(k) = \mathrm{Re}\,F(k) + i\,\mathrm{Im}\,F(k) = \sum_{j = 0}^{\lfloor n/2\rfloor}(-1)^{j}c_{2j}+i\,\sum_{j=0}^{\lfloor n/2\rfloor}(-1)^{j}c_{2j+1} =0.
\end{equation}
 %----------------%
Were the resonances real, that is, embedded eigenvalues, they would be given by the conditions $\mathrm{Re}\,F(k) = 0\,\land\, \mathrm{Im}\,F(k) = 0$. While those provide constraints on the signed sums of the coefficients $c_j(k)$ only, in the graph of spectral $L$-dependence any particular real value $k$ satisfying $F(k) = 0$ is bounded from one side by the crossings between the curves $\tilde{F}(k) = 0$ and $\tan{kL} = 0$, as at these points all the $c_j(k)$ except $c_0(k)$ must be decreasing to zero to account for the divergence of $\cot^j{kL}$ terms, and from the other by crossings between the curves $\tilde{F}(k) = 0$ and $\cot{kL} = 0$, because then necessarily $c_0 =0$. Should these two points be close to each other with respect to their position on the momentum axis $k$, they form nearly horizontal lines in the $L$-dependence graph, finally indicating the resonance energy if they occur at the same value of $k$.

The graph may not have real resonances, of course, and in fact the non-real ones are those we are interested in, however, we observe a similar behavior if our $k$ is sufficiently close to the real axis. The inverse $F^{-1}$ is a meromorphic function and the resonances are its poles in the lower complex half-plane. The trace of the function on the real axis is given by
 %----------------%
\begin{equation*}
    \mathrm{Re}\,F^{-1}(k) = \frac{\mathrm{Re}\,F(k)}{|\mathrm{Re}\,F(k)|^2+|\mathrm{Im}\,F(k)|^2}, \quad k\in\R,
\end{equation*}
 %----------------%
which suggests to look for the real part of the resonance positions at zeros of $\mathrm{Re}\,F(\cdot)$. The trouble is that in general the real parts of $F(\cdot)$ zeros do not coincide with real zeros of $\mathrm{Re}\,F(\cdot)$. This is not a big obstacle, however, because significant resonances are only those in which the pole term of $F^{-1}(\cdot)$ dominates, typically when the pole is close to the real axis. Let $k_0=\delta-i\eta$ be a zero of $F(\cdot)$, then Taylor expansion of $F$ at $k_0$ gives
 %----------------%
\begin{equation*}
    \mathrm{Re}\,F^{-1}(k) = c\,\frac{k-\delta}{(k-\delta)^2+\eta^2}\, \big(1+\OO(((k-\delta)^2+\eta^2)^{-2}\big)
\end{equation*}
 %----------------%
with an appropriate constant $c$, hence for such resonances the zeros of $\mathrm{Re}\,F(\cdot)$, which contain $c_0(\cdot)$, determine the real part of the poles in the leading order, and the values of $\mathrm{Im}\,F(\cdot)$, composed of the odd-numbered $c_j(\cdot)$, characterize the resonance widths.

Let us illustrate these claims in a particular example, namely the cross-shaped resonator considered in the previous section; for definiteness we choose the crossbar length $l=1$ and strength parameter value $\alpha = 1$. Note that the system has interesting features, for instance it exhibits a nontrivial quantum holonomy \cite{CT09} with respect to adiabatic changes of parameter $\lambda$ as a particular resonance trajectory in Figure~\ref{Resonance_travel} shows, cf.~\cite{EL10}. Figures~\ref{Resonances_lambdas_1_2} and \ref{Resonances_lambdas_2_2} show numerical solutions of the spectral condition \eqref{Cross_cond_spec} as functions of the cut-off length $L$. For several values of $\lambda$ the system has embedded eigenvalues, in particular, for $\lambda=1$ equation \eqref{Cross_cond_spec} is solved by $k = m\pi/2$, $m \in \mathbb{Z}$. We thus expect sharp resonances in the vicinity of this point, and it is indeed what we see in the top left picture of Figure~\ref{Resonances_lambdas_1_2} referring to $\lambda = 0.95$, at least for several lowest values of $|m|$; higher in the spectrum the picture becomes smeared. The spectra are computed for seven other values of $\lambda$ corresponding to resonances indicated in Figure~\ref{Resonance_travel}. It is apparent how this resonance widens and shrinks with respect to the pole distance from the real axis. Finally, in the lower right picture of Figure~\ref{Resonances_lambdas_2_2}, calculated for $\lambda = 0.03$, we see resonances in the vicinity of the points $k = m\pi$, $m \in \mathbb{Z}$, which are again the analytic solutions to $\eqref{Cross_cond_spec}$ with $\lambda = 0$.

 %----------------%
	\begin{figure}[b]\centering
		\includegraphics[width = \textwidth]{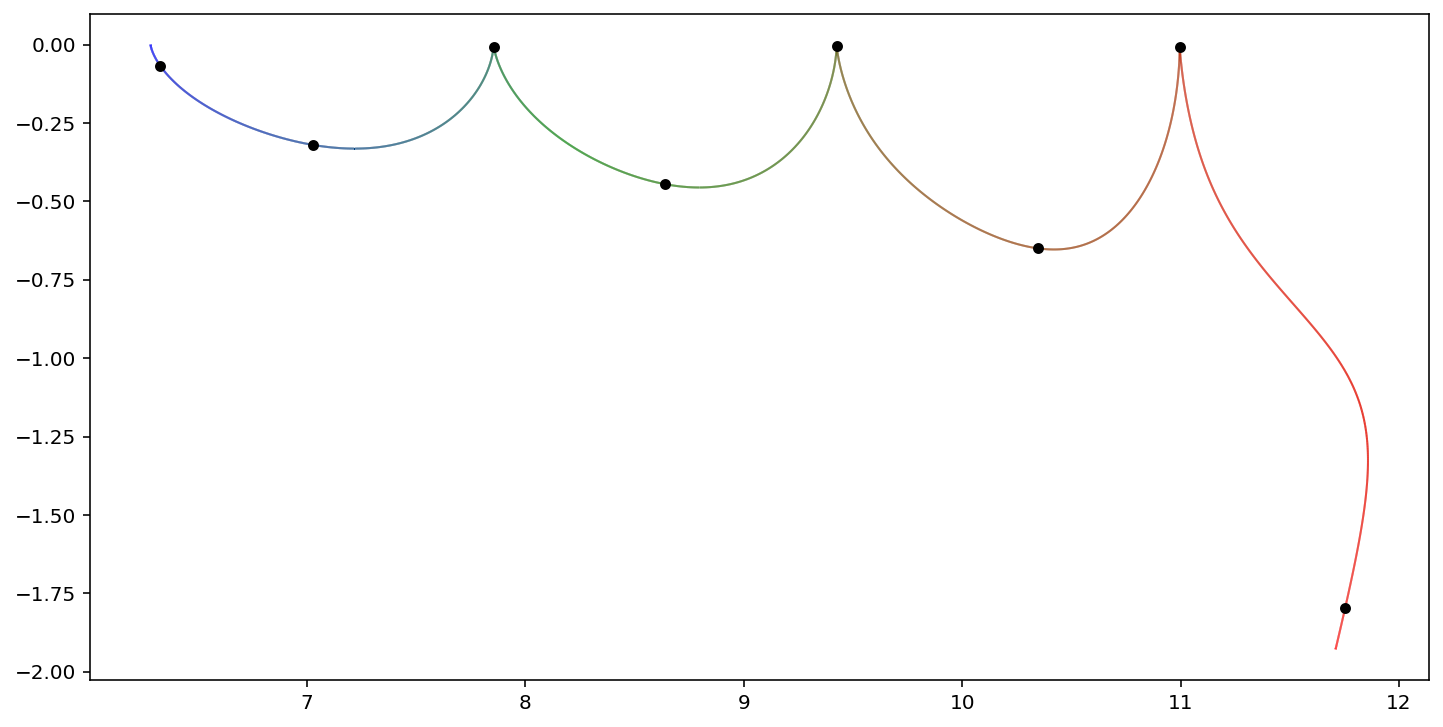}
		\caption{The trajectory of the resonance pole in the lower complex half-plane starting at $k = 2\pi$ for the cross-shaped resonator with $\alpha = 1$. The color-coding (visible online) shows the dependence on $\lambda$ changing from blue ($\lambda = 1$) to red ($\lambda = 0$), the black points on the curve correspond to the values of $\lambda$ used in Figure~\ref{Resonances_lambdas_1_2} and \ref{Resonances_lambdas_2_2}.}
		\label{Resonance_travel}
	\end{figure}
 %----------------%

 %----------------%
	\begin{figure}[b]\centering
		\includegraphics[width = \textwidth]{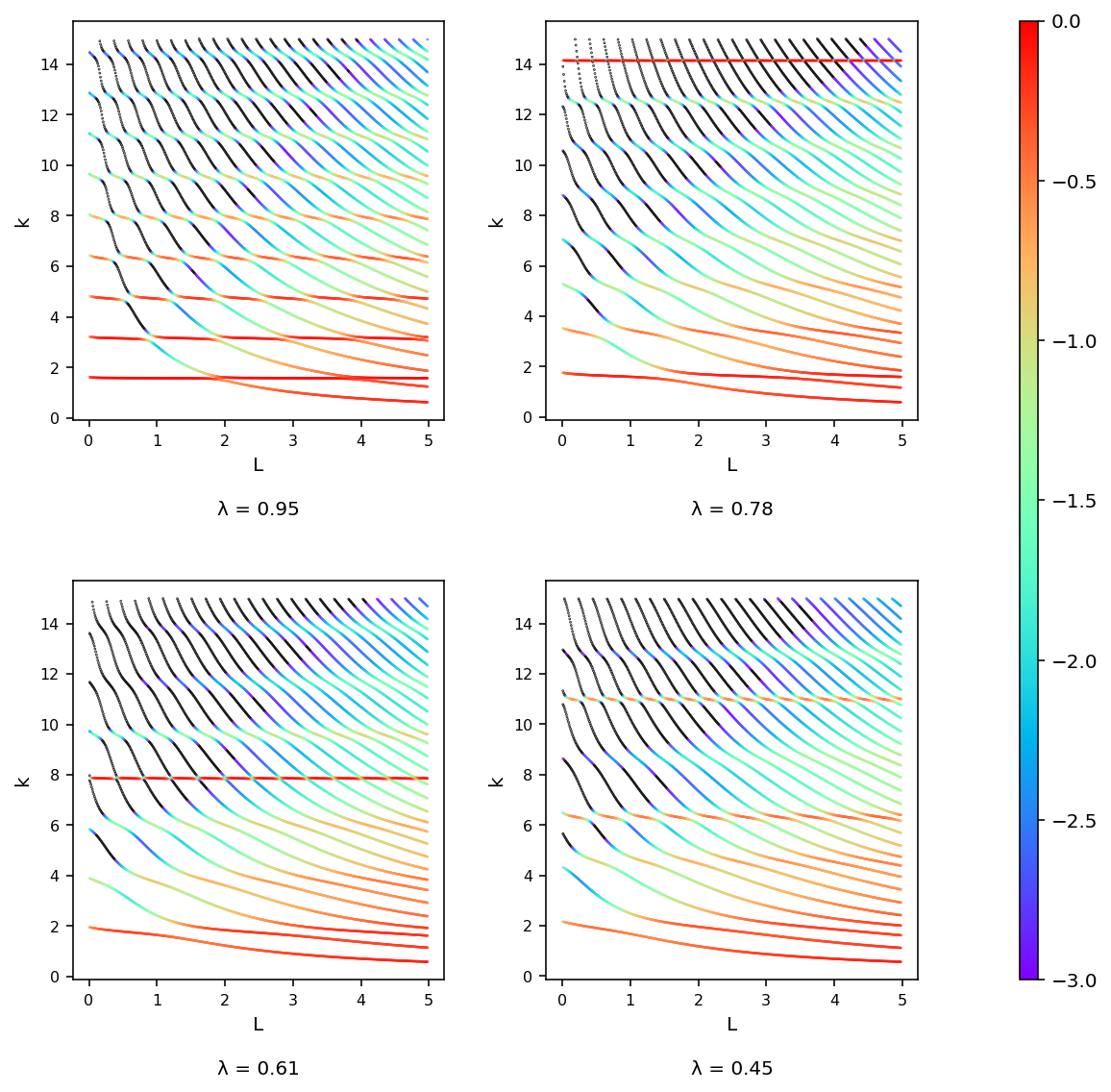}
		\caption{The numerical solutions of \eqref{Cross_cond_spec} with $\alpha= 1$ for specific values of $\lambda$. The curves are colored (visible online) according to the value of derivative -- running between $-3$ and zero -- to highlight the resonances; points with derivative outside of the specified interval are colored black.}
		\label{Resonances_lambdas_1_2}
	\end{figure}
 %----------------%

 %----------------%
	\begin{figure}[b]\centering
		\includegraphics[width = \textwidth]{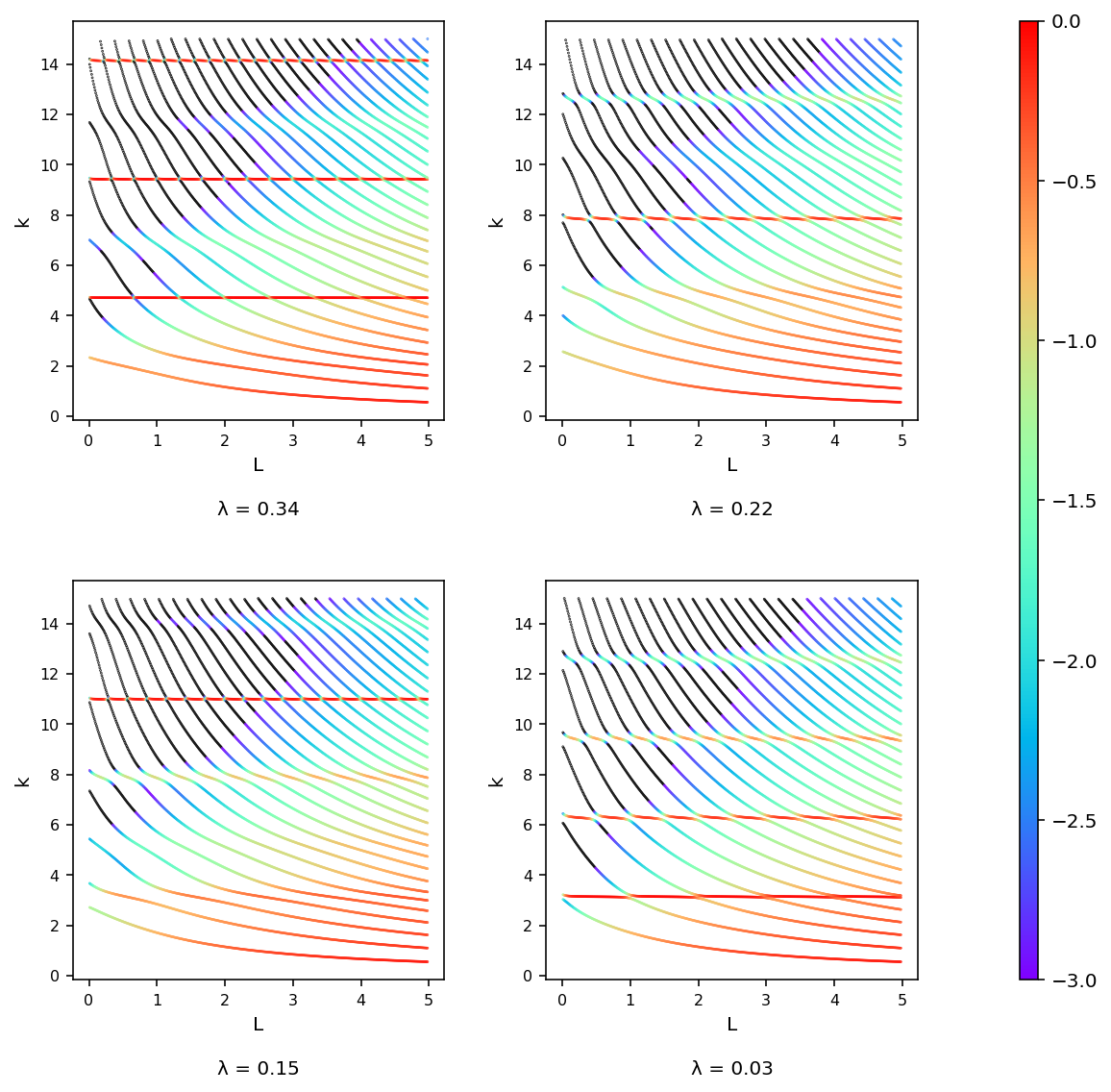}
		\caption{The numerical solutions of \eqref{Cross_cond_spec} with $\alpha= 1$ for specific values of $\lambda$. The color coding is the same as in the previous figure.}
		\label{Resonances_lambdas_2_2}
	\end{figure}
 %-------------

To illustrate that values of resonant energies indeed lie between crossings of the spectral curves with the curves $\cot{kL}=0$ and $\tan{kL} = 0$, consider again the resonance at $\lambda = 0.95$. It is easy to solve the equation for $\mathrm{Im}\,F(k) = 0$, which in this specific example coincides with the condition $\tan{kL} = 0$ as the resonance condition consists only of the terms $c_0(k)$ and $c_1(k)$ contributing to the $\mathrm{Re}\,F(k)$ and $\mathrm{Im}\,F(k)$ respectively: it has two roots,
 %----------------%
 \begin{equation*}
         k = \frac{m\pi}{1+\lambda} \; \text{and} \; k = \frac{m\pi}{1-\lambda}, \; m\in \mathbb{Z}.
 \end{equation*}
 %----------------%
Only the first is visible in Figure~\ref{Resonance_width_0_95} where it is indicated by the horizontal black dotted lines. On the other hand, the requirement $\mathrm{Re}\,F(k)=0$ gives $2k\sin{2k} +\cos{2k\lambda}-\cos{2k} = 0$ which has to be solved numerically; the result is indicated by the horizontal black dashed lines in Figure~\ref{Resonance_width_0_95}.

%----------------%
	\begin{figure}[b]\centering
		\includegraphics[width = \textwidth]{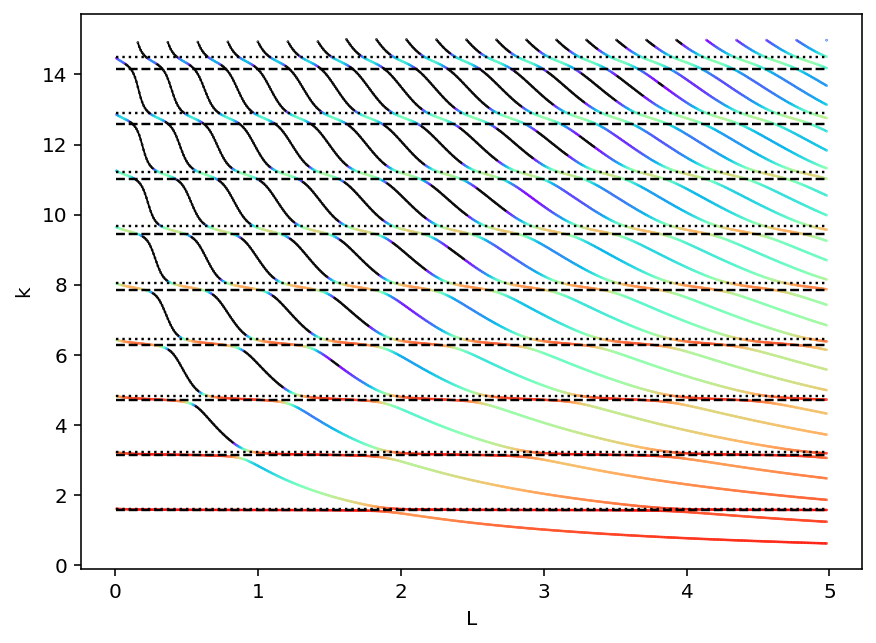}
		\caption{The numerical solutions of \eqref{Cross_cond_spec} for $\lambda = 0.95$. The color coding (visible online) is the same as in Figure~\ref{Resonances_lambdas_1_2}}.
		\label{Resonance_width_0_95}
	\end{figure}
 %----------------%

%%%%%%%%%%%%%%%
\section*{Acknowledgements}
%%%%%%%%%%%%%%%
Thanks are due to the QMath16 organizers for distinguishing the preliminary results of this paper presented at the conference by Jan Peka\v{r}. J.L. acknowledges financial support by the Faculty of Science, University of Hradec Kr\'{a}lov\'{e}.

%\newpage
%\def\cprime{$'$}

\end{document}